\documentclass[]{interact}

\usepackage{epstopdf}
\usepackage[caption=false]{subfig}

\usepackage[longnamesfirst,sort]{natbib}
\bibpunct[, ]{(}{)}{;}{a}{,}{,}

\usepackage{graphicx}
\DeclareGraphicsExtensions{.pdf,.png,.jpg,.eps}

\usepackage{amsmath}
\usepackage{amssymb}
\usepackage{natbib}

\theoremstyle{plain}
\newtheorem{theorem}{Theorem}[section]

\newtheorem{corollary}[theorem]{Corollary}

\theoremstyle{definition}

\theoremstyle{remark}
\newtheorem{remark}{Remark}

\begin{document}

\articletype{ARTICLE TEMPLATE}

\title{Stability and Instability Divergence Conditions for Dynamical Systems}

\author{
\name{I.~B. Furtat\textsuperscript{a,b}\thanks{CONTACT I.~B. Furtat. Email: cainenash@mail.ru}}
\affil{\textsuperscript{a} Institute for Problems of Mechanical Engineering Russian Academy of Sciences, 61 Bolshoy ave V.O., St.-Petersburg, 199178, Russia;}
{\textsuperscript{b} ITMO University, 49 Kronverkskiy ave, Saint Petersburg, 197101, Russia;}
}

\maketitle

\begin{abstract}
A novel method for stability and instability study of autonomous dynamical systems using the flow and divergence of the vector field is proposed.
A relation between the method of Lyapunov functions and the proposed method is established. 
Bendixon and Bendixon-Dulac theorems for $n$th dimensional systems are extended.
Based on the proposed method, the state feedback control law is designed. 
The control signal is obtained from the partial differential inequality.
The examples illustrate the application of the proposed method and the existing ones.
\end{abstract}


\section{Introduction}
Dynamical models describe many processes in the surrounding macro and microcosm. One of the important problems is the study of the solution convergence of these models. However, it is not always possible to find an explicit solution, but numerical solutions can significantly differ from the exact ones, see, e.g. \cite{Atkinson09}.

It is well known that the method of Lyapunov functions allows one to study the stability of solutions of differential equations without solving them. This method is proposed by A.M. Lyapunov at the end of the 19 century in his doctoral dissertation with application to problems of astronomy and fluid motion. Depending on the problem being solved, Lyapunov function is also interpreted as a potential function (\cite{Yuan14}), an energy function (\cite{Bikdash00}) or a storage function (\cite{Willems72}). However, the main restriction of the method of Lyapunov functions is to find these functions.

Methods for stability study of dynamical systems based on the divergence of a vector field are alternative to the method of Lyapunov functions. The first fundamental results based on divergent stability conditions were proposed in \cite{Zaremba54,Fronteau65,Brauchli68}. 
The important results for investigation of system stability were proposed by A. Rantzer, A.A. Shestakov,  A.N. Stepanov and V.P. Zhukov.
In \cite{Jukov78} the instability problem of nonlinear systems using the divergence of a vector field is considered. 
In \cite{Shestakov78,Jukov79} a necessary condition for the stability of nonlinear systems in the form of non-positivity of the vector field divergence is proposed. 
First, an auxiliary scalar function is introduced in \cite{Shestakov78,Jukov90} to study the instability of nonlinear systems. 
However, the similar scalar function is considered in \cite{Krasnoselski63} for stability and instability study of dynamical systems, but using the method of Lyapunov functions.
In \cite{Shestakov78,Jukov99} stability conditions for second-order systems are obtained. 
Then in \cite{Rantzer00,Rantzer01} the convergence of almost all solutions of arbitrary order nonlinear dynamical systems is considered. 
As in \cite{Shestakov78,Jukov90,Jukov99} the auxiliary scalar function (density function) is used for the stability study of dynamical models. Additionally, in \cite{Rantzer00,Rantzer01} the synthesis of the control law based on divergence conditions is proposed.
The auxiliary functions in \cite{Shestakov78,Jukov99,Rantzer00,Rantzer01} are similar except their properties at the equilibrium point. 
Currently, method from \cite{Rantzer00,Rantzer01} has been extended to various systems, see i.e. \cite{Monzon03,Loizou08,Castaneda15,Karabacak18}.

However, the necessary condition is sufficiently rough in \cite{Shestakov78,Jukov99}. 
The sufficient condition stability is proposed only for second-order systems in \cite{Jukov99}. 
Theorem 1 in \cite{Rantzer01} guarantees the convergence of almost all solutions, but not all solutions. 
Proposition 2 in \cite{Rantzer01} allows to study the asymptotic stability, but proposition conditions have sufficient restriction. 

Recently, a new divergent method for stability study of autonomous dynamical system is proposed in \cite{Furtat20a, Furtat20b}. 
This method allows one to overcome the above mentioned problems. 
In the present paper we extend results \cite{Furtat20a, Furtat20b} such that new sufficient stability and instability conditions are obtained for linear and nonlinear systems. These results allows one to extend Bendixon and Bendixon-Dulac theorems to $n$th dimensional systems and design the new control laws.


The paper is organized as follows. 
Section \ref{Sec2} contains new necessary and sufficient stability and instability conditions, the extension of Bendixon and Bendixon-Dulac theorems, as well as, the numerical examples and comparisons with the methods from \cite{Shestakov78,Jukov99,Rantzer00,Rantzer01}. 
Section \ref{Sec3} describes methods for design the state feedback control law and numerical examples. 
Finally, Section \ref{Sec4} collects some conclusions.

\textit{Notations}. In the paper the following notation are used: $\textup{grad}\{W(x)\}=\Big[\frac{\partial W}{\partial x_1}, ...,\frac{\partial W}{\partial x_n}\Big]^{\rm T} $ is the gradient of the scalar function $W(x)$, $\textup{div}\{h(x)\}=\frac{\partial h_1}{\partial x_1}+...+\frac{\partial h_n}{\partial x_n}$ is the divergence of the vector field $h(x)=[h_1(x),...,h(x)_n]^{\rm T}$, $|\cdot|$ is the Euclidean norm of the corresponding vector.
We mean that the zero equilibrium point is stable (unstable) if it is Lyapunov stable (unstable) (\cite{Khalil09}).

\section{Maun results} \label{Sec2}

\subsection{Stability of nonlinear systems} \label{Sec2a}

Consider a dynamical system in the form
\begin{equation}
\label{eq1}
\begin{array}{l} 
\dot{x}=f(x),
\end{array}
\end{equation} 
where $x=[x_1(t), ..., x_n(t)]^{\rm T}$ is the state vector, $f=[f_1,...,f_n]^{\rm T}: D \to \mathbb R^{n}$ is the continuously differentiable function in $D \subset \mathbb R^{n}$. The set $D$ contains the origin and $f(0)=0$.
For simplicity, we assume that the domain of attraction $D_A$ of the point $x=0$ coincides with the domain $D$. However, all obtained results is valid if $D_A \subset D$ or $D_A=\mathbb R^{n} $. Denote by $\bar{D}$ a boundary of the domain $D$.

Let us formulate the necessary stability condition for system \eqref{eq1}.

\begin{theorem}
\label{Th1}
Let $x=0$ be an asymptotically stable equilibrium point of \eqref{eq1}. Then there exists a positive definite continuously differentiable function $S(x)$ such that $S(x)\to\infty$ for $x \to \bar{D} $, $|\textup{grad}\{S(x)\}| \neq 0$ for any $x \in D \setminus \{0\} $ and at least one of the following conditions holds:

\textup{1)} the function $\textup{div}\{|\textup{grad} \{S(x)\} |f(x)\}$ is integrable in the domain $V=\{x \in D: S(x) \leq C \}$ and $\int_{V} \textup{div} \{|\textup{grad} \{S(x)\} |f(x)\}dV <0$ for all $C>0$;

\textup{2)} the function $\textup{div} \{|\textup{grad} \{S^{-1}(x)\}|f(x)\}$ is integrable in the domain $V_{inv}=\{x \in D: S^{-1}(x) \geq C \}$ and $\int_{V_{inv}} \textup{div}\{|\textup{grad}\{S^{-1}(x)\} |f(x)\} dV_{inv}>0$ for all $C>0$.
\end{theorem}

\begin{proof}
According to \cite[Theorem 4.17]{Khalil09} if $x=0$ is an asymptotically stable equilibrium point of system \eqref{eq1}, then there exists a continuously differentiable positive definite function $S(x)$ such that $S(x) \to \infty$ for $x \to \bar{D}$, $\textup {grad} \{S(x)\}^{\rm T} f(x)<0$ for any $x \in D \setminus \{0\}$ and $\textup {grad} \{S(x)\}^{\rm T}f(x)\Big|_{x = 0}=0$. If $D=\mathbb R^n$, then the function $S(x)$ is radially unbounded.
Next, we consider two cases separately which correspond to the functions $S(x)$ and $S^{-1}(x)$.
 
1. If $\textup{grad}\{S(x)\}^{\rm T}f(x)<0 $, then
$\frac{1}{|\textup{grad}\{S(x)\}|} \textup{grad}\{S(x)\}^{\rm T}| \textup{grad} \{S(x)\}|f(x)<0$. Therefore, the following expression holds
$F_1=\oint_{\Gamma}\frac{1}{|\textup{grad}\{S(x)\}|}\textup{grad}\{S(x)\}^{\rm T}|\textup{grad}\{S(x)\}|f(x)d \Gamma<0.$
Using Divergence theorem (or Gauss theorem), we get $F_1=\int_{V}\textup{div}\{|\textup{grad}\{S(x)\}|f(x)\}dV<0$.

2. If $\textup{grad}\{S(x)\}^{\rm T}f(x)<0$, 
then $\textup{grad}\{S^{-1}(x)\}^{\rm T}f(x)=-S^{-2}(x)\textup {grad} \{S (x)\}^{\rm T} f(x)>0$. 
On the other hand,
$\textup {grad} \{S^{-1}(x)\}^{\rm T}f(x)
= \frac{1}{|\textup {grad}\{S^{-1}(x)\}|} \textup{grad}\{S^{- 1} (x)\}^{\rm T} |\textup{grad}\{S^{- 1}(x)\}|f(x).$
Therefore, the following relation is satisfied
$F_2=\oint_{\Gamma_{inv}}\frac{1}{|\textup{grad}\{S^{-1}(x)\}|}
\textup{grad}\{S^{-1}(x)\}^{\rm T}|\textup{grad}\{S^{-1}(x)\}|f(x)d \Gamma_{inv}>0.$
According to Divergence theorem, we get
$F_2=\int_{V_{inv}}\textup{div}\{|\textup{grad}\{S^{-1}(x)\}|f(x)\}dV_{inv}>0.$
Theorem \ref{Th1} is proved.
\end{proof}

The integrals in Theorem \ref{Th1} explicitly depend on the function $S(x)$ that depends on the integration surface. Let us formulate a corollary that weakens this requirement.

\begin{corollary}
\label{Th2}
Let $x =0$ be the asymptotically stable equilibrium point of system \eqref{eq1}. Then there exist positive definite continuously differentiable functions 
$\phi(x)$ and $S(x)$ such that $\phi(x) \to \infty$ and $S(x) \to \infty$ for $x \to \bar{D}$, $|\textup{grad} \{S(x)\} |\neq 0$ for any $x \in D \setminus \{0\}$  and at least one of the following conditions holds:

1) the function $\textup {div}\{\rho(x)f(x)\}$ is integrable in the domain $V=\{x \in D: S(x) \leq C \}$ and $\int_{V} \textup{div} \{\rho(x) f(x) \} dV <0$ for all $C>0 $, where $\rho(x)=\phi(x) |\textup{grad}\{S(x)\}|$;

2) the function $\textup{div}\{\rho^{-1}(x)f(x)\}$ is integrable in the domain $V_{inv}=\{x \in D: S^{-1}(x \geq C\}$ and $\int_{V_{inv}} \textup{div} \{\rho^{-1}(x)f(x)\} dV_{inv}>0 $ for all $C>0$, where $\rho^{-1}(x)=\phi^{-1} (x)| \textup{grad} \{S^{-1}(x)\}|$.

\end{corollary}

\begin{proof}
Following the proof of Theorem \ref{Th1}, consider two cases.

1. If $\textup{grad} \{S(x)\}^{\rm T}f(x)<0$, then $\phi(x) \textup{grad} \{S(x)\}^{\rm T} f(x)<0 $. Therefore, the further proof is similar to the proof of Theorem \ref{Th1}, but tacking into account the flow of the vector field $\phi(x)|\textup{grad}\{S(x)\}|f(x)$ through the surface $\Gamma$.

2. If $\textup{grad}\{S(x)\}^{\rm T}f(x)<0$, then $\phi^{-1}(x) \textup{grad} \{S(x) \}^{\rm T} f(x)<0$. Therefore, the further proof is similar to the proof of Theorem \ref{Th1}, but taking into account the flow of the vector field 
$\phi^{-1}(x)|\textup{grad} \{S^{-1}(x)\}|f(x)$ through the surface $\Gamma_{inv}$. The corollary is proved.
\end{proof}

\begin{remark}
\label{rem2}
If the function $\rho(x)$ is chosen such that $\textup{div} \{\rho(x) f(x) \}$ and $\textup{div} \{\rho^{-1}(x)f(x)\}$ 
are integrable, as well as, $\textup{div} \{\rho(x)f(x)\}<0$ and $\textup{div} \{\rho^{-1}(x)f(x)\}>0$ for any $x \in D \setminus \{0\}$, 
then the corresponding conditions $ \int_{V} \textup{div}\{\rho(x)f(x)\} dV<0$ and $\int_{V_{inv}} \textup{div} \{\rho^{-1}(x) f(x) \}dV_{inv}>0$ in Corollary \ref{Th2} are satisfied. In \cite{Rantzer01} the integrability of $\textup{div} \{\rho^{-1}(x)f(x)\}$ and the condition $\textup{div} \{\rho^{-1}(x)f(x)\}>0$ are required only for convergence of almost all solutions of \eqref{eq1}. Thus, the results of \cite{Rantzer01} are special case in Corollary \ref{Th2}.
\end{remark}

Now let us formulate a sufficient condition of stability of \eqref{eq1}.

\begin{theorem}
\label{Th2a}
Let $\rho(x)$ be a positive definite continuously differentiable function in $D$. The equilibrium point $x=0$ of system \eqref{eq1} is stable (asymptotically stable) if at least one of the following conditions holds:

1) $\textup {div} \{\rho(x) f(x)\} \leq \rho(x) \textup{div} \{f (x)\}$ 
$(\textup{div} \{\rho(x) f(x)\} < \rho(x) \textup{div} \{f(x)\})$ for any $x \in D \setminus \{0\} $ and $\textup{div} \{\rho(x)f(x)\} \big|_{x=0}=0$;

2) $\textup{div} \{\rho^{-1}(x)f(x)\} \geq 0$ $(\textup{div} \{\rho^{-1}(x) f (x) \}>0) $ and $\textup{div} \{f(x)\} \leq 0$ for any $x \in D \setminus \{0\} $ and 
$\lim_{|x| \to 0} \big [\rho^2(x) \textup{div} \{\rho^{-1}(x)f(x)\} \big]=0$;

3) $\textup{div}\{\rho(x) f(x)\} \leq \beta (x) \rho^2 (x) \textup{div} \{\rho^{-1}(x) f(x)\}$ 
$(\textup{div} \{\rho(x)f(x)\} < \beta(x) \rho^2(x) \textup{div}\{\rho^{-1}(x)f(x)\})$, where 
$\beta(x) > 1$ and $\textup{div}\{f(x)\} \leq 0$ or only $\beta(x)=1$ for any $x \in D \setminus \{0\}$, as well as, 
$\textup{div} \{\rho(x)f(x)\} \big |_{x = 0}=0$ and $\lim_{|x| \to 0} \big [\rho(x) \textup{div} \{\rho^{-1}(x)f (x)\} \big]=0$;

4) $\textup{div}\{\rho(x) f(x)\} \leq 0$  and $\textup{div}\{\rho^{-1}(x) f(x)\} \geq 0$ $(\textup{div}\{\rho(x) f(x)\} < 0$  and $\textup{div}\{\rho^{-1}(x) f(x)\} > 0)$, as well as, 
$\textup{div}\{\rho(x) f(x)\} \big |_{x = 0}=0$ and $\lim_{|x| \to 0} \big [\rho(x) \textup{div} \{\rho^{-1}(x)f (x)\} \big]=0$.
\end{theorem}

\begin{proof}
Consider the proof for each case separately. The proof of asymptotic stability is omitted because it is similar to the proof of stability, but taking into account the sign of a strict inequality.

1. From the relation 
$\textup{div} \{\rho(x) f(x)\} = \textup{grad}\{\rho(x)\}^{\rm T} f(x)+\textup{div}\{f(x)\} \rho(x)$ 
implies that if 
$\textup{div}\{\rho(x)f(x)\} \leq \textup{div}\{f(x)\}\rho(x)$, 
then $\textup{grad}\{\rho(x)\}f(x) \leq 0 $ 
in the domain $D \setminus \{0\}$. 
Consider the condition $\rho(0)=0$. 
If $\textup{div}\{\rho(x)f(x)\} \big |_{x=0}=0$, 
then $\textup{grad}\{\rho (x)\}f(x)\big |_{x=0}=0$. 
Therefore, according to Lyapunov theorem (\cite{Khalil09}), system \eqref{eq1} is stable.

2. From the expression 
$\textup{div}\{\rho^{-1}(x)f(x)\}=\textup{grad}\{\rho^{-1} (x)\}^{ \rm T} f(x)+\textup{div}\{f(x)\}\rho^{-1}(x)$ it 
follows that 
$\textup{grad}\{\rho(x)\}^{\rm T} f(x)=\rho(x)\textup{div}\{f(x)\}-\rho^{2}(x) \textup{div}\{\rho^{-1}(x)f(x)\}$. 
If $\textup{div}\{\rho^{-1}(x)f(x)\} \geq 0$ and $\textup{div} \{f(x)\} \leq 0$, 
then $\textup{grad}\{\rho(x)\}^{\rm T} f(x) \leq 0$ in $D \setminus \{0\}$. If $\lim_{|x| \to 0} \big [\rho^2 (x) \textup{div} \{\rho^{-1}(x)f(x)\} \big]=0$, then $\lim_{|x| \to 0} \big[\textup{grad}\{\rho(x)\}f(x)\big]=0$. Therefore, system \eqref{eq1} is stable.

3. Condition 3 is a combination of conditions 1 and 2. Summing $\beta(x) \textup {grad}\{\rho(x)\}^{\rm T} f(x)=\beta(x) \rho(x)\textup{div} \{f(x)\}-\beta(x) \rho^{2}(x) \textup {div} \{\rho^{-1}(x)f(x)\}$ 
and 
$\textup{grad}\{\rho(x)\}^{\rm T}f(x)=
\textup{div}\{\rho(x)f(x)\}-\textup{div}\{f(x)\}\rho(x)$, 
we get $(1+\beta(x)) \textup{grad}\{\rho(x)\}^{\rm T}f(x)=\textup{div}\{\rho(x)f (x)\}-\beta(x)\rho^{2}(x) \textup{div}\{\rho^{-1}(x)f(x)\}+(\beta(x)-1)\rho(x)\textup{div}\{f (x)\}$. 
If $\textup{div}\{\rho(x)f(x)\} \leq \beta(x) \rho^2(x) \textup{div}\{\rho^{-1} (x)f(x)\}$ 
for $\beta(x)=1 $ or $\beta(x)>1$ and $\textup{div} \{f (x)\} \leq 0$, then 
$\textup{grad}\{\rho(x)\}^{\rm T}f(x) \leq 0$ in the region $D \setminus \{0\}$. 
If $\textup {div}\{\rho(x)f(x)\} \big|_{x=0}=0$ and 
$\lim_{|x| \to 0} \big [\rho^2(x) \textup{div} \{\rho^{-1}(x)f(x)\} \big]=0$, 
then 
$\lim_{|x| \to 0} \big [\textup{grad} \{\rho(x)\}f(x)\big]=0$. Therefore, system \eqref{eq1} is stable.

4. From the relation 
$2\textup{grad}\{\rho (x)\}^{ \rm T} f(x)=-\rho^2(x)\textup{div}\{\rho^{-1}(x)f(x)\}+\textup{div}\{f(x)\}\rho(x)+\textup{grad}\{\rho (x)\}^{ \rm T} f(x)$ it follows that $2\textup{grad}\{\rho (x)\}^{ \rm T} f(x)=-\rho^2(x)\textup{div}\{\rho^{-1}(x)f(x)\}+\textup{div}\{\rho(x)f(x)\}$.
If $\textup{div}\{\rho(x)f(x)\} \leq 0$ and $\textup{div}\{\rho^{-1}(x)f(x)\} \geq 0$, then $\textup{grad}\{\rho (x)\}^{ \rm T} f(x) \leq 0$ in $D \setminus \{0\}$. Therefore, system \eqref{eq1} is stable.
Theorem \ref{Th2a} is proved.
\end{proof}

It is noted in Introduction that the result of \cite{Shestakov78,Jukov99} is applicable only to second-order systems. Next, we consider an illustration of the proposed results for third-order systems and compare the results with ones from \cite{Rantzer01}.

\textit{Example 1}.
Consider the system

\begin{equation}
\label{eq6}
\begin{array}{l} 
\dot{x}_1=x_2-2x_1 x_3^2,
\\
\dot{x}_2=-x_1-2x_2x_3^2,
\\
\dot{x}_3=-2x_3^3,
\end{array}
\end{equation} 
which has an equilibrium point $(0,0,0)$.

Choose $\rho(x)=|x|^{2\alpha}$, where $\alpha$  is a positive integer.
First, verify the conditions of Corollary \ref{Th2}. Since $\textup{div}\{\rho(x)f(x)\}=-|x|^{2\alpha}(4\alpha+10)x_3^2<0$ for any $\alpha$ and $x_3 \neq 0$, as well as, $\textup{div}\{\rho^{-1}(x)f (x)\}=(4\alpha-10)x_3^2|x|^{-2\alpha}>0$ for $\alpha \geq 3$ and $x_3 \neq 0$, then the conditions of Corollary \ref{Th2} are satisfied. 
Since the function $\textup{div}\{\rho^{-1}(x)f(x)\} $ is integrable in $\{x \in \mathbb R^n: |x| \geq 1 \}$, then the conditions of Theorem 1 in \cite{Rantzer01} (convergence of almost all solutions of \eqref{eq6}) are satisfied too.

Now let us verify the conditions of Theorem \ref{Th2a}. The relation 
$\textup{div}\{\rho(x)f(x)\}-\rho(x)\textup{div}\{f(x)\}=-4\alpha x_3^2|x|^{2\alpha}<0$
 holds for any $\alpha$ and $x_3 \neq 0$.
In turn, $\textup{div}\{f(x)\}=-10x_3^2<0$ and the function $\textup{div}\{\rho^{-1}(x)f(x)\}>0$ for any $\alpha \geq 3$ and $x_3 \neq 0$ (this conclusion can also be obtained using Proposition 2 in \cite{Rantzer01}). 
Let $\beta(x)=\beta \geq 1$. Then $\textup{div}\{\rho(x)f(x)\}-\beta\rho^2(x) \textup{div}\{\rho^{-1}(x)f(x)\}=-(4\alpha+10+4\beta\alpha-10\beta)x_3^2|x|^{2\alpha}<0$ for $\alpha>\frac{5(\beta-1)}{2(\beta+1)}$ and $x_3 \neq 0$. 
The conditions $\textup{div}\{\rho(x) f(x)\} < 0$  and $\textup{div}\{\rho^{-1}(x) f(x)\} > 0$ hold for $\alpha \geq 3$ and $x_3 \neq 0$.
All four cases gave the same results. 
Therefore, system \eqref{eq6} is asymptotically stable with any initial conditions when $x_3(0) \neq 0$. If the initial conditions contain $x_3 (0)=0$, then system 
\eqref{eq6} is stable. The phase trajectories of \eqref{eq6} are shown in Fig.~\ref{Fig5}, where the cycle is obtained for the initial condition with $x_3=0$, the spiral is obtained for $x_3 \neq 0$.

Thus, Corollary \ref{Th2} and Theorem \ref{Th2a}, as well as the results of \cite{Rantzer01}, give positive answers about the stability of \eqref{eq6}. Additionally, the conditions of Theorem \ref{Th2a} allow establishing when system \eqref{eq6} is stable and when it is asymptotically stable.

\begin{figure}[h!]
\center{\includegraphics[width=0.7\linewidth]{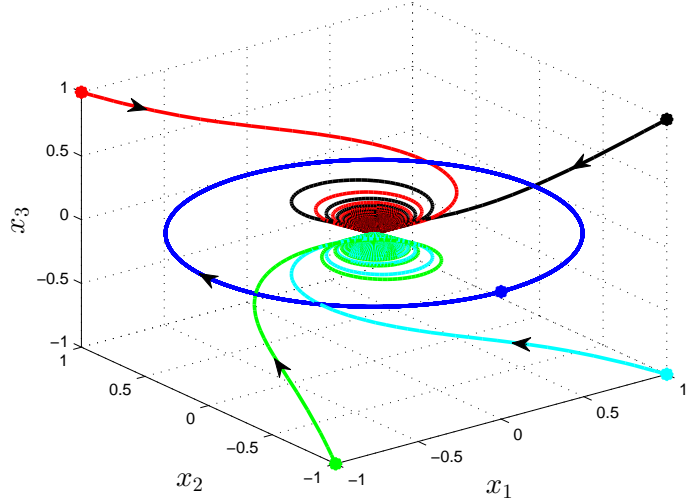}}
\caption{Phase trajectories of system \eqref{eq6}.}
\label{Fig5}
\end{figure}

\textit{Example 2}.
Consider the system

\begin{equation}
\label{eq04aa}
\begin{array}{l} 
\dot{x}_1=-x_1+x_1^2-x_2^2-x_3^2,
\\
\dot{x}_2=-x_2+2x_1 x_2,
\\
\dot{x}_3=-x_3+2x_1 x_3,
\end{array}
\end{equation} 
which has two equilibrium points $(0,0,0)$ and $(1,0,0)$. 
All trajectories of the system converge to the point $(0,0,0)$, except those that start on the semi-axis $x_1 \geq 1$, $x_2=0 $ and $x_3=0$ (see Fig.~\ref{Ust_Rantz_3rd}). 
Choose $\rho(x)=|x|^{2\alpha}$, $\alpha$ is a positive integer. Then $\textup{div}\{\rho^{-1}(x)f(x)\}=|x|^{-2 \alpha} [2\alpha-3+2x_1(3-\alpha)]>0 $ for $\alpha=3$. The function $\textup{div} \{f(x)\}=-3+6x_1$ does not satisfy the condition $\textup{div}\{f(x)\} \leq 0$ for $x_1>0.5$. 
Relations $\textup{div}\{\rho(x)f(x)\} \leq \rho(x) \textup{div}\{f(x)\}$, $\textup{div}\{\rho(x)f(x)\} \leq \beta(x) \rho^2(x) \textup{div}\{\rho^{-1}(x)f(x)\} $ and  $\textup{div}\{\rho(x)f(x)\} \leq 0$ are not satisfied too. 
As a result, the conditions of Corollary \ref{Th2} (and the conditions of Theorem 1 in \cite{Rantzer01}) are fulfilled in this example, but the conditions of Theorem \ref{Th2a} (and the conditions of Proposition 2 in \cite{Rantzer01} are not satisfied.

\begin{figure}[h!]
\center{\includegraphics[width=0.7\linewidth]{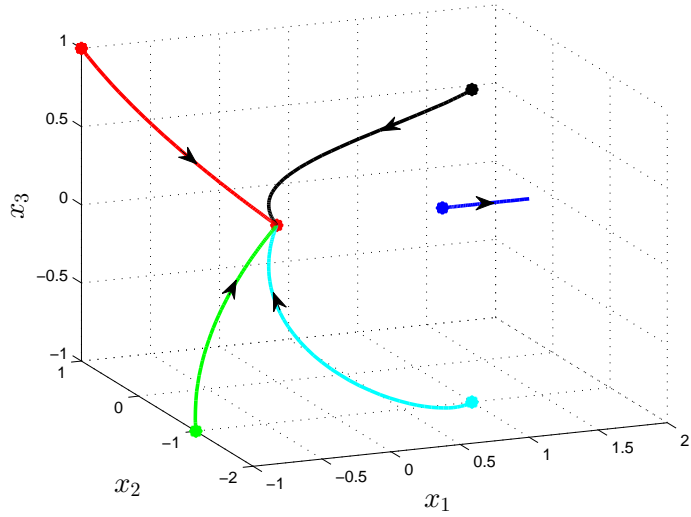}}
\caption{Phase trajectories of system \eqref{eq04aa} with two equilibrium points.}
\label{Ust_Rantz_3rd}
\end{figure}

\textit{Example 3.}
Consider the system

\begin{equation}
\label{eq6aa}
\begin{array}{l} 
\dot{x}_1=-4 x_1 x_2^2 - x_1^3,
\\
\dot{x}_2=4x_1^2 x_2-x_2^3-8x_2 x_3^2,
\\
\dot{x}_3=-x_3^3+8x_2^2x_3
\end{array}
\end{equation}
with equilibrium point $(0,0,0)$. The phase trajectories of \eqref{eq6aa} are shown in Fig.~\ref{Fig5aa} for various initial conditions.

\begin{figure}[h!]
\center{\includegraphics[width=0.7\linewidth]{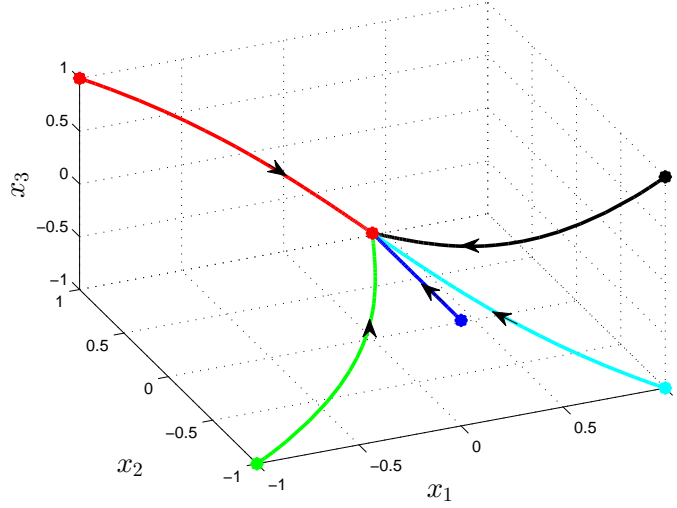}}
\caption{Phase trajectories of system \eqref{eq6aa}.}
\label{Fig5aa}
\end{figure}

Choose $\rho(x)=|x|^{2\alpha}$, $\alpha$  is a positive integer and verify the conditions of Corollary \ref{Th2}.
Considering $\textup{div}\{\rho(x)f(x)\}=|x|^{2\alpha-2}[(-2\alpha+1) x_1^4+(-2\alpha+1)x_2^4+(-2\alpha-11)x_3^4+2x_1^2x_2^2-10x_1^2 x_3^2-10x_2^2 x_3^2]$,
we get $\int_V \textup{div} \{\rho(x)f(x)\}dV<0$ for any $C$ and $\alpha$.
For $\textup {div}\{\rho^{-1}(x)f(x)\}=|x|^{-2\alpha-2}[(2\alpha+1)x_1^4+(2\alpha+1)x_2^4+(2\alpha-11)x_3^4+2x_1^2x_2^2-10x_1^2x_3^2-10x_2^2x_3^2$
the condition $\int_{V_{inv}}\textup{div}\{\rho^{-1}(x)f(x)\}dV_{inv}>0$ holds for any $C$ and $\alpha \geq 3$.
Consequently, the conditions of Corollary \ref{Th2} are satisfied (the conditions of Theorem 1 in \cite{Rantzer01} are satisfied only for $\alpha \geq 8$).

Verify the conditions of Theorem \ref{Th2a}. 
The relation
$\textup{div}\{\rho(x)f (x)\}-\rho(x)\textup{div}\{f(x)\}=-2\alpha|x|^{2\alpha-2}(x_1^4+x_2^4+ x_3^4)<0$
holds for any $\alpha$ and $x \neq 0$.
The function $\textup {div} \{f(x)\}=x_1^2+x_2^2-11x_3^2$ is not positive definite. Thus, Proposition 2 in \cite{Rantzer01} and the second case of Theorem \ref{Th2a} cannot be satisfied.
Condition $\textup{div}\{\rho(x)f(x)\}-\beta\rho^2(x)\textup{div}\{\rho^{-1}(x)f(x)\}<0$ in Theorem \ref{Th2a} holds for $\beta=1$ and $x \neq 0$. The relations $\textup{div}\{\rho(x)f(x)\} \leq 0$ and $\textup{div}\{\rho^{-1}(x)f(x)\} \geq 0$ are satisfied for for $\alpha \geq 6$ and $x \neq 0$.

As a result, the conditions of Corollary \ref{Th2} and Theorem \ref{Th2a} are satisfied for system \eqref{eq6aa}. Thus, $(0,0,0)$ is an asymptotically stable equilibrium point. According to \cite{Rantzer01}, we can only conclude that almost all solutions of \eqref{eq6aa} converge to $(0,0,0)$, because the conditions of Proposition 2 in \cite{Rantzer01} are not satisfied and only the conditions of Theorem 1 in \cite{Rantzer01} hold.

\subsection{Stability of linear systems} \label{Sec2b}

\begin{theorem}
\label{Th2a_lin}
Given $\alpha>0$. The linear system $\dot{x}=Ax$, $x \in \mathbb R^n$ is stable if at least one of the following conditions holds:

1) $A^{\rm T}P+PA-\frac{1}{\alpha} \frac{\beta-1}{\beta+1} trace(A)P<0$ for $\beta = 1$ or for $\beta > 1$ and $trace(A) \leq 0$;

2) $A^{\rm T}P+PA-\frac{1}{\alpha}trace(A)P<0$ and $A^{\rm T}P+PA+\frac{1}{\alpha}trace(A)P<0$.
\end{theorem}

\begin{proof}
Let $\rho(x)=(x^{\rm T}Px)^{\alpha}$. 
According to Theorem \ref{Th2a} (case 3), the relation $\textup{div}\{\rho(x)f(x)\}-\beta \rho^2(x) \textup{div}\{\rho^{-1}(x)f(x)\} = \alpha(1+\beta)(x^{\rm T}Px)^{\alpha-1}x^{\rm T}[A^{\rm T}P+PA-\frac{1}{\alpha} \frac{\beta-1}{\beta+1} trace(A)P]x < 0$ is satisfied, if $A^{\rm T}P+PA-\frac{1}{\alpha} \frac{\beta-1}{\beta+1} trace(A)P<0$ holds for $\beta = 1$ or for $\beta > 1$ and $trace(A) \leq 0$. 

Considering Theorem \ref{Th2a} (case 4), the relations $\textup{div}\{\rho(x)f(x)\} = \alpha(x^{\rm T}Px)^{\alpha-1}x^{\rm T}[A^{\rm T}P+PA+\frac{1}{\alpha} trace(A)P]x < 0$ and $\textup{div}\{\rho^{-1}(x)f(x)\}  = -\alpha(x^{\rm T}Px)^{-\alpha-1}x^{\rm T}[A^{\rm T}P+PA-\frac{1}{\alpha}  trace(A)P]x > 0$ are satisfied, if $A^{\rm T}P+PA+trace(A)P<0$ and $A^{\rm T}P+PA-trace(A)P<0$ simultaneously hold. Theorem \ref{Th2a_lin} is proved.
\end{proof}

As a result, the matrix inequality in Theorem \ref{Th2a_lin} (case 1) simultaneously includes Lyapunov inequality (for $\beta = 1$) and inequality from \cite{Rantzer01} (for $\beta > 1$). In Theorem \ref{Th2a_lin} (case 2) the matrix inequality $A^{\rm T}P+PA-trace(A)P<0$ from \cite{Rantzer01} is complemented by new inequality $A^{\rm T}P+PA+trace(A)P>0$. The sum of the inequalities from Theorem \ref{Th2a_lin} (case 2) gives Lyapunov inequality.

\subsection{Instability conditions} \label{Sec2c}

The instability condition in the form $\textup{div}\{f(x)\}>0$ in $x \in D \setminus \{0\}$ for system \eqref{eq1} is considered in \cite{Zaremba54,Jukov78}. 
In \cite{Jukov79} the new condition $\textup{div}\{\rho(x) f(x)\}>0$ in $x \in D \setminus \{0\}$ allows one to extended a class of  investigated systems by introducing the positive definite continuously differentiable function $\rho(x)$. 
Differently from \cite{Zaremba54,Jukov78}, we propose new results that allows one to further extend the class of investigated systems by using a continuously differentiable function $\rho(x)$, that can be not a positive definite, and consider only a part of the set $D \setminus \{0\}$.

\begin{theorem}
\label{Th2a_ins}
Let $\rho(x): D \to \mathbb R$ be a continuously differentiable function such that $\rho(0)=0$ and $\rho(x_0)>0$ for some $x_0$ with arbitrary small $\|x_0\|$. Define a sets $U=\{x \in B: \rho(x)>0 \}$ and $B=\{x \in \mathbb R^n: \|x\| \leq r, r>0 \}$ such that $B \subseteq D$. The equilibrium point $x=0$  of system \eqref{eq1} is unstable, if at least one of the following conditions holds  for any $x \in U$:

1) $\textup {div} \{\rho(x) f(x)\} > \rho(x) \textup{div} \{f (x)\}$;

2) $\textup{div} \{\rho^{-1}(x)f(x)\} < 0$ and $\textup{div} \{f(x)\} \geq 0$;

3) $\textup{div}\{\rho(x) f(x)\} > \beta (x) \rho^2 (x) \textup{div} \{\rho^{-1}(x) f(x)\}$, where 
$\beta(x) > 1$ and $\textup{div}\{f(x)\} \geq 0$ or only $\beta(x)=1$;

4) $\textup{div}\{\rho(x) f(x)\}> 0$  and $\textup{div}\{\rho^{-1}(x) f(x)\} < 0$.
\end{theorem}

\begin{proof}
According to Chetaev theorem \cite{Khalil09}, if $\textup{grad}\{\rho(x)\}f(x) > 0$  for $ x \in U$, then the equilibrium point $x=0$ is unstable. Also, the proof of Theorem \ref{Th2a_ins} is based on the proof of Theorem \ref{Th2a}.

From the relation 
$\textup{div} \{\rho(x) f(x)\} = \textup{grad}\{\rho(x)\}^{\rm T} f(x)+\textup{div}\{f(x)\} \rho(x)$ 
implies that if 
$\textup{div}\{\rho(x)f(x)\} > \textup{div}\{f(x)\}\rho(x)$  for $ x \in U$ (see case 1), 
then $\textup{grad}\{\rho(x)\}f(x) > 0$ 
for $ x \in U$. 

Consider
$\textup{grad}\{\rho(x)\}^{\rm T} f(x)=\rho(x)\textup{div}\{f(x)\}-\rho^{2}(x) \textup{div}\{\rho^{-1}(x)f(x)\}$. 
If $\textup{div}\{\rho^{-1}(x)f(x)\} < 0$ and $\textup{div} \{f(x)\} > 0$ (see case 2), 
then $\textup{grad}\{\rho(x)\}^{\rm T} f(x) > 0$ for $ x \in U$. 

Consider $(1+\beta(x)) \textup{grad}\{\rho(x)\}^{\rm T}f(x)=\textup{div}\{\rho(x)f (x)\}-\beta(x)\rho^{2}(x) \textup{div}\{\rho^{-1}(x)f(x)\}+(\beta(x)-1)\rho(x)\textup{div}\{f (x)\}$. 
If $\textup{div}\{\rho(x)f(x)\} > \beta(x) \rho^2(x) \textup{div}\{\rho^{-1} (x)f(x)\}$ 
for $\beta(x)=1 $ or $\beta(x)>1$ and $\textup{div} \{f (x)\} > 0$   for $ x \in U$ (see case 3), then 
$\textup{grad}\{\rho(x)\}^{\rm T}f(x) > 0$  for $ x \in U$. 

Consider $2\textup{grad}\{\rho (x)\}^{ \rm T} f(x)=-\rho^2(x)\textup{div}\{\rho^{-1}(x)f(x)\}+\textup{div}\{\rho(x)f(x)\}$.
If $\textup{div}\{\rho(x)f(x)\} > 0$ and $\textup{div}\{\rho^{-1}(x)f(x)\} < 0$   for $ x \in U$ (see case 4), then $\textup{grad}\{\rho (x)\}^{ \rm T} f(x) > 0$  for $ x \in U$. 
Theorem \ref{Th2a} is proved.
\end{proof}

\textit{Example 4.}  Consider the system

\begin{equation}
\label{eq_Ex_4}
\begin{array}{l} 
\dot{x}_1=x_1+g_1(x),
\\
\dot{x}_2=-x_2+g_2(x),
\end{array}
\end{equation}
where $|g_1(x)| \leq k_1 \|x\|^2$ and $|g_2(x)| \leq k_1 \|x\|^2$, $k_1>0$, see  \cite{Khalil09}. Let $\rho(x)=0.5(x_1^2-x_2^2) > 0$. Therefore, $(0,0)$ is the equilibrium point. Differently from \cite{Zaremba54,Jukov78,Jukov79}, the chosen function $\rho(x)$ is not positive definite.

Consider case 1 of Theorem \ref{Th2a_ins}. The condition $\textup {div} \{\rho(x) f(x)\} - \rho(x) \textup{div} \{f (x)\} \geq \|x\|^2(1-2k_1\|x\|)>0$ holds for $r<\frac{1}{2k_1}$.

Consider case 2 of Theorem \ref{Th2a_ins}. The relation $\textup {div} \{\rho^{-1}(x) f(x)\} \geq 
\frac{2}{(x_1^2-x_2^2)^2} \Big[ -\|x\|^2(1-2k_1\|x\|)+0.5(x_1^2-x_2^2)\left(\frac{\partial g_1}{\partial x_1}+\frac{\partial g_2}{\partial x_2}\right)\Big] \geq 
-\frac{2}{(x_1^2-x_2^2)^2} \|x\|^2 \left(1-2k_1\|x\|+0.5k_2\|x\| \right)$ holds for $0 \leq \frac{\partial g_1}{\partial x_1}+\frac{\partial g_2}{\partial x_2} \leq k_2 \|x\|$, $k_2>4k_1$, and $r<\frac{2}{k_2-4k_1}$. Thus, differently from case 1, now we need the additional restriction on the derivatives $\frac{\partial g_1}{\partial x_1}$ and $\frac{\partial g_2}{\partial x_2}$.

In case 3 of Theorem \ref{Th2a_ins} the condition $\textup{div}\{\rho(x) f(x)\} > \beta (x) \rho^2 (x) \textup{div} \{\rho^{-1}(x) f(x)\}$ holds only for $\beta=1$ and the same conditions as in case 2.

Consider case 4 of Theorem \ref{Th2a_ins}. The relations $\textup{div}\{\rho(x) f(x)\}> 0 \geq 
\|x\|^2(1-2k_1\|x\|-0.5k_2\|x\|) > 0$ and $\textup {div} \{\rho^{-1}(x) f(x)\} <0$ hold for the same conditions as in case 2. 
Thus, all four cases show that the point $(0,0)$ is unstable.

\subsection{Extension of Bendixson and Bendixson-Dulac theorems for $n$th dimensional systems} \label{Sec2d}

The lack of periodic solutions in a simply connected domain in $\mathbb R^2$ can be established by  Bendixson and Bendixson-Dulac theorems (\cite{Guckenheimer83}). The next two theorems are extensions of Bendixson and Bendixson-Dulac theorems to the $n$th dimensional systems.


\begin{theorem}
\label{Th4}
Let $D \subseteq \mathbb R^n$ is a simply connected domain. If $\textup{div}\{f(x)\}$ does not change the sign for all $x \in D$ (except possibly in a set of measure $0$), then system \eqref{eq1} has no invariant closed subset with a positive measure in $D$.
\end{theorem}

\begin{proof}
Denote $\Gamma$ is the closed invariant subset with a positive measure in $D$, $\bar{\Gamma}$ is the boundary of $\Gamma$, $int\{\Gamma\}$ is the interior of $\Gamma$, $V = \bar{\Gamma} \cup int \{\Gamma\}$, and $\bar{V}$ is a volume of $V$.
According to Liouville's theorem and \cite[Theorem 1  in p. 69]{Arnold14}, $d\bar{V}/dt=\int_{V}\textup{div}\{f(x)\} d V=0$, i.e. the volume  of a closed invariant subset does not change at any time $t$. 
If $\textup{div}\{f(x)\}$ does not change the sign in $D$, then the value of $d\bar{V}/dt$ has negative or positive value, i.e. the volume $\bar{V}$ is decreased or increased. 
We have a contradiction except possibly in a set of measure $0$, where $\textup{div}\{f(x)\}$ can be zero. 
Thus, system \eqref{eq1} has no invariant closed subset with a positive measure in $D$. Theorem \ref{Th4} is proved.
\end{proof}

\begin{theorem}
\label{Th5}
Let $D \subseteq \mathbb R^n$ is a simply connected domain. 
If there exists the continuously differentiable function $\rho(x)$ such that $\textup{div}\{\rho(x) f(x)\}$ does not change the sign for all $x \in D$ (except possibly in a set of measure $0$), then system \eqref{eq1} has no invariant closed subset  with a positive measure in $D$.
\end{theorem}

\begin{proof}
Let $\Gamma=\{x \in D: S(x) = C \}$ be an invariant closed subset, $\bar{V}$ is the phase volume of $V=\{x \in D: S(x) \leq C \}$ and $\rho(x)=\phi(x) |\textup{grad}\{S(x)\}|$, where $\phi(x)$  is a continuously differentiable function. 
If $\textup{div}\{\rho(x)f(x)\} \neq 0$, then $\int_{V}\textup{div}\{\rho(x)f(x)\}dV \neq 0$. 
Using Divergence theorem as well as the proof of Corollary \ref{Th2}, we have $\int_{V}\textup{div}\{\rho(x)f(x)\}dV = \oint_{\Gamma}\phi(x) \textup{grad}\{S(x)\}^{\rm T} f(x) d \Gamma \neq 0$. 
If system \eqref{eq1} has an invariant closed subset  with a positive measure, then $\textup{grad}\{S(x)\}^{\rm T} f(x)=0$ and $\oint_{\Gamma}\phi(x) \textup{grad}\{S(x)\}^{\rm T} f(x) d \Gamma = 0$. We have a contradiction  except possibly in a set of measure $0$, where $\textup{grad}\{S(x)\}^{\rm T} f(x)$ can be zero. Thus, system \eqref{eq1} has no invariant closed subset with a positive measure in $D$. Theorem \ref{Th5} is proved.
\end{proof}

\section{Control law design}
\label{Sec3}

Consider a dynamical system in the form

\begin{equation}
\label{eq9}
\begin{array}{l} 
\dot{x}=\xi(x)+g(x)u(x),
\end{array}
\end{equation} 
where $u(x)$ is the control signal, the functions $\xi(x)$, $g(x)$ and $u(x)$ are continuously differentiable in $D$, $\xi(0)=0$, $g(0)=0$ and system \eqref{eq9} is controllable in $D$.

\begin{theorem}
\label{Th6}
Let $\rho(x)$ be a positive definite continuously differentiable function  in 
$x \in D$. The closed-loop system is stable (asymptotically stable) if the control law $u(x)$ is chosen such that at least one of the following conditions holds:

1) $\textup{div}\{\rho(x)(\xi(x)+g(x)u(x))\} \leq \rho(x) \textup{div}\{\xi(x)+g(x)u(x)\}$ $(\textup{div} \{\rho(x) (\xi(x)+g(x)u(x))\}<\rho(x)\textup{div} \{\xi(x)+g(x)u(x)\})$ for any $x \in D \setminus \{0\}$ and $\textup{div} \{\rho(x)(\xi(x)+g(x)u(x))\}\big|_{x=0}=0$;

2) $ \textup{div}\{\rho^{-1}(x)(\xi(x)+g(x)u(x))\} \geq 0$ $(\textup{div}\{\rho^{-1}(x)(\xi(x)+g(x)u(x))\}>0)$ for any $x \in D \setminus \{0\}$ and 
$\lim_{|x| \to 0} \big [\rho^2(x)\textup{div}\{\rho^{-1}(x)(\xi(x)+g(x)u(x))\}\big]=0$;

3) $\textup{div}\{\rho(x)(\xi(x)+g(x)u(x))\} \leq \beta(x) \rho^2(x)\textup{div} \{\rho^{- 1}(x) (\xi(x)+g(x)u(x))\}$, $\beta \geq 1$ $(\textup{div}\{\rho(x)(\xi (x)+g(x)u(x))\}
<\beta(x)\rho^2(x)\textup{div}\{\rho^{-1}(x)(\xi(x)+g(x)u(x))\})$, 
where $\beta(x)>1$ and $\textup{div}\{\xi(x)+g(x)u(x)\} \leq 0$ or only 
$\beta(x)=1$ for any $x \in D \setminus \{0\}$, as well as, 
$\textup {div} \{\rho(x) (\xi(x)+g(x)u(x))\} \big|_{x=0}=0$ and 
$\lim_{|x| \to 0} \big[\rho(x) \textup{div}\{\rho^{-1}(x)(\xi(x)+g(x)u(x))\} \big]=0$.

4) $\textup{div}\{\rho(x) (\xi(x)+g(x)u(x))\} \leq 0$  and $\textup{div}\{\rho^{-1}(x) (\xi(x)+g(x)u(x))\} \geq 0$ $(\textup{div}\{\rho(x) (\xi(x)+g(x)u(x))\} < 0$  and $\textup{div}\{\rho^{-1}(x) (\xi(x)+g(x)u(x))\} > 0)$, as well as, 
$\textup{div}\{\rho(x)(\xi(x)+g(x)u(x))\} \big |_{x = 0}=0$ and $\lim_{|x| \to 0} \big [\rho(x) \textup{div} \{\rho^{-1}(x)(\xi(x)+g(x)u(x))\} \big]=0$.

\end{theorem}

Since system \eqref{eq9} is controllable in $D$, the proof of Theorem \ref{Th6}  is similar to the proof of Theorem \ref{Th2a} (denoting by $f(x)=\xi(x)+g(x)u(x)$).

\begin{remark}
If the control law design is based on the method of Lyapunov functions, then it is required to solve the algebraic inequality $\textup{grad}\{V\}(f+gu)<0$.  According to Theorem \ref{Th6}, the control law is chosen from the feasibility of differential inequality. This gives new opportunities for the control law design.
\end{remark}

\textit{Example 5.}
Consider the system

\begin{equation}
\label{eq7ex}
\begin{array}{l} 
\dot{x}_1=dx_2 - x_1 x_2^2,
\\
\dot{x}_2=u,
\end{array}
\end{equation}
where $d$ takes the values of $0$ and $1$.
It is required to design the control law $u$ that ensures the asymptotic stability of \eqref{eq7ex}. Obviously, system \eqref{eq7ex} is not asymptotically stable for $u=0$ and for any values of $d$.
Choose $\rho(x)=|x|^{2\alpha}$, $\alpha$ is a positive integer and use the third case of Theorem \ref{Th6}.

1. Let $d=0$. Compute $\textup{div}\{\rho(x)(\xi(x)+g(x)u(x))\}-\beta(x)\rho^2(x)\textup{div}\{\rho^{-1}(x)(\xi(x)+g(x)u(x))\}=-2\alpha(1+\beta)x_1^2x_2^2+2\alpha (1+\beta) ux_2+(1-\beta)(-x_2^2+\frac{\partial u}{\partial x_2})(x_1^2+x_2^2)$.
Choosing $u=-x_2^3$, we get $\textup{div}\{\rho(x)(\xi(x)+g(x)u(x))\}-\beta(x)\rho^2(x)\textup{div}\{\rho^{-1}(x)(\xi(x)+g(x)u(x))\}<0$ for $\beta \geq 1$, $\alpha>\frac{2(\beta-1)}{\beta+1}$ and $x_2 \neq 0$, as well as, $\textup{div}\{\xi(x)+g(x)u(x)\} \leq 0$. The phase trajectories of the closed-loop system are shown in Fig.~\ref{Fig7aa},\textit{a}.

2. Let $d=1$. Compute $\textup{div}\{\rho(x)(\xi(x)+g(x)u(x))\}-\beta(x)\rho^2(x) \textup{div}\{\rho^{-1}(x)(\xi(x)+g(x)u(x))\}=2\alpha(1+\beta)x_1x_2-2\alpha(1+ \beta) x_1^2 x_2^2 + 2\alpha(1+\beta) ux_2 + (1-\beta) \left(-x_2^2 + \frac{\partial u}{\partial x_2} \right) (x_1^2 + x_2^2)$.
Choosing $u= -x_1-(\beta-1)x_2^3$, we get $\textup{div}\{\rho(x)(\xi(x)+g(x) u(x))\}-\beta(x)\rho^2(x)\textup{div}\{\rho^{-1}(x)(\xi(x)+g(x)u(x))\}<0$ for 
$\beta \geq 1$ and $\alpha>\max\left\{\frac{(\beta-1)(3\beta-2)}{2(\beta+1)}, \frac{3\beta-2}{2(\beta+1)}\right\}$, as well as, $\textup{div}\{\xi(x)+g(x)u(x)\} \leq 0$.
The phase trajectories of the closed-loop system are shown in Fig.~\ref{Fig7aa}, \textit{b} for $\beta=2$.

\begin{figure}[h!]
\begin{minipage}[h]{1\linewidth}
\center{\includegraphics[width=0.7\linewidth]{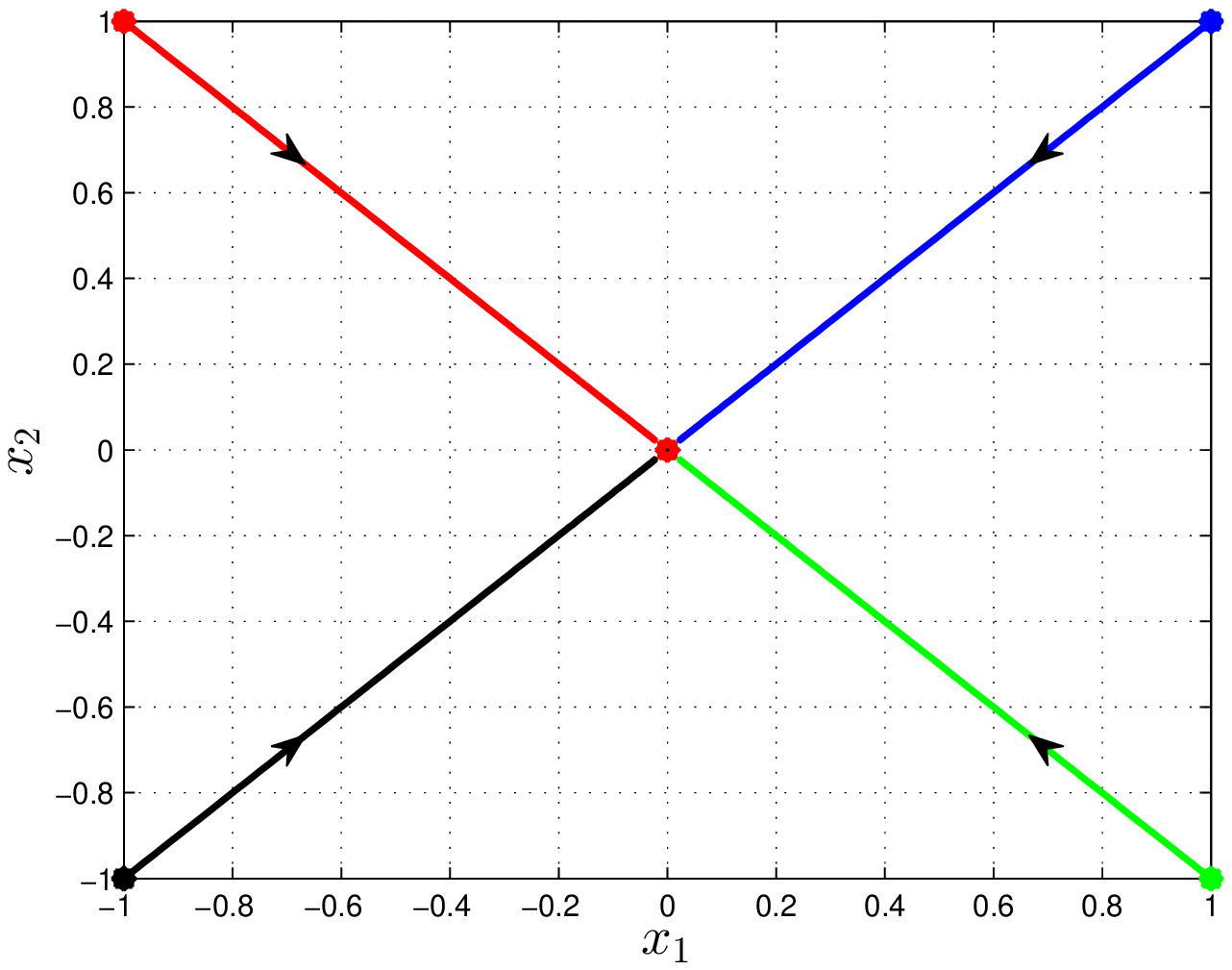}} \textit{a} \\ 
\end{minipage}
\vfill
\begin{minipage}[h]{1\linewidth}
\center{\includegraphics[width=0.7\linewidth]{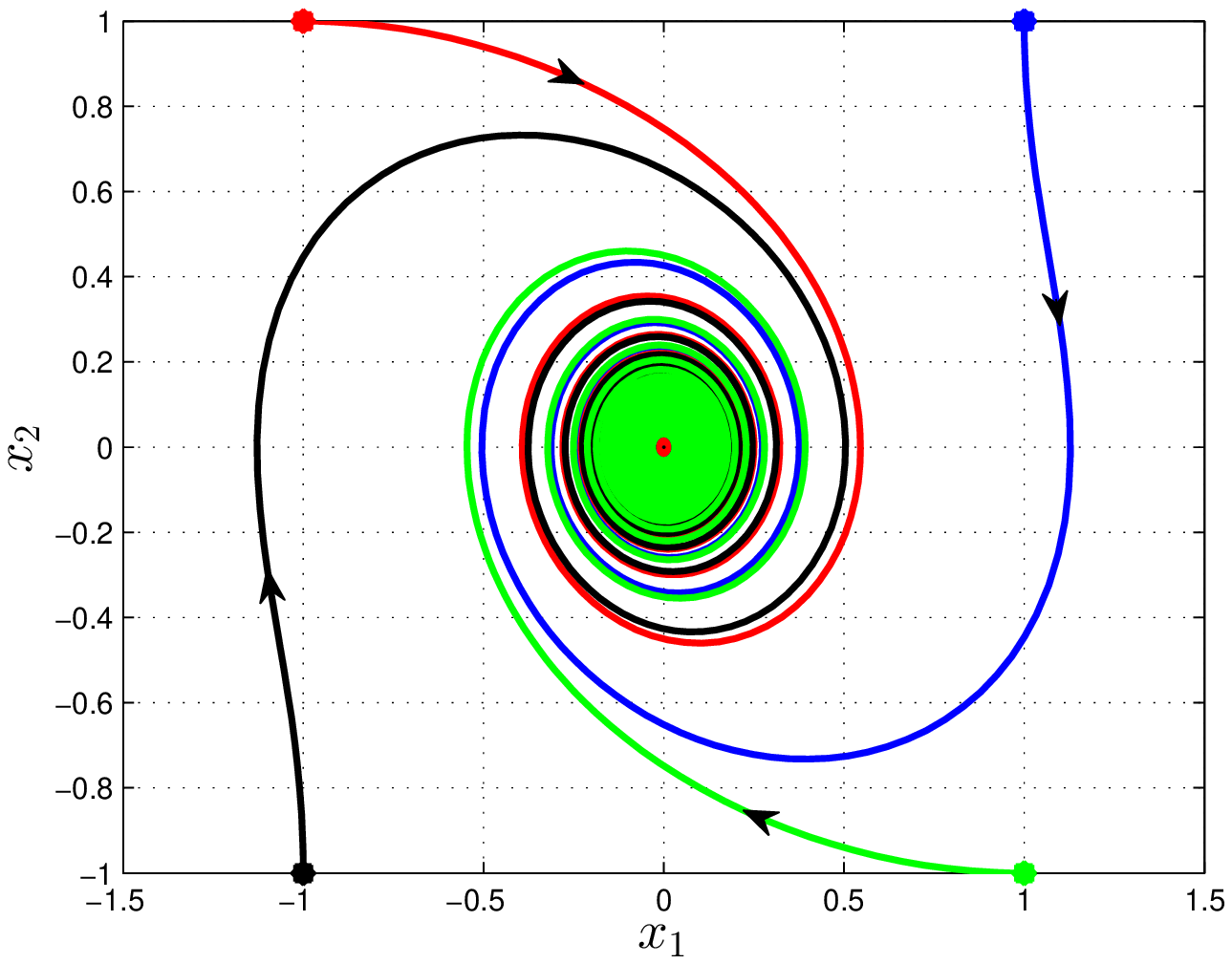}} \textit{b} \\
\end{minipage}
\caption{The phase trajectories in the closed-loop system for $d=0$ (\textit{a}) and for $d=1$, $\beta=2$ (\textit{b}).}
\label{Fig7aa}
\end{figure}

\section{Conclusion}
\label{Sec4}

A method for stability  and instability study of dynamical systems using the properties of the flow and divergence of the vector field is proposed. 
To study the stability and instability, the existence of a certain type of integration surface or the existence of an auxiliary scalar function is required. 
Necessary and sufficient stability and instability conditions are proposed. 
The extension of Bendixon and Bendixon-Dulac theorems for $n$th dimensional systems is given.

The obtained results are applied to synthesis the feedback control law for dynamical systems. 
It is shown that the control law is found as a solution of a differential inequality, while the control law based on the method of Lyapunov functions is found as a solution of an algebraic inequality.

\section{Acknowledgment}
The results of Section 3 were developed under support of RSF (grant 18-79-10104) in IPME RAS. The other researches were partially supported by grants of Russian Foundation for Basic Research No. 19-08-00246 and Government of Russian Federation, Grant 074-U01.

\end{document}